\documentclass[a4paper,11pt,preprint]{article}

\usepackage{amsmath}
\usepackage{amsthm}
\usepackage{amssymb}

\usepackage{algorithmic}
\usepackage{algorithm}
\usepackage{xcolor}
\usepackage{graphicx}

\usepackage{wrapfig}

\usepackage{multicol}
\usepackage{units}

\usepackage{cite}

\usepackage{paralist}

\usepackage[left=1.0in,top=1.0in,right=1.0in,bottom=1.0in]{geometry}

\theoremstyle{definition}

\newtheorem{definition}{Definition}
\newtheorem{theorem}{Theorem}
\newtheorem{corollary}[theorem]{Corollary}
\newtheorem{lemma}[theorem]{Lemma}
\newtheorem{prop}{Proposition}

\newenvironment{ap_lemma}[1]{\par\noindent{\bf Lemma~#1.} \em}{}
\newenvironment{ap_theorem}[1]{\par\noindent{\bf Theorem~#1.} \em}{}

\graphicspath{{figure/}}

\def\MC#1{{\mathcal #1}}
\def\MBB#1{{\mathbb #1}}
\def\MB#1{{\mathbf #1}}

\newcommand{\BIGLR}[3]{{\left#1#3\right#2}}
\newcommand{\BIGP}[1]{{\BIGLR{(}{)}{#1}}}

\newcommand{\BIGBP}[1]{{\BIGLR{\{}{\}}{#1}}}

\newcommand{\BIGC}[1]{{\BIGLR{|}{|}{#1}}}

\newcommand{\OP}[1]{{\operatorname{#1}}}

\title{Optimal Time-Convex Hull under the $L_p$ Metrics}

\author{
Bang-Sin Dai \\[6pt] \small Dep. of Computer Science \\ \small and Information Engineering \\ \small National Taiwan Univ. \\ \small Taiwan \\[2pt] \footnotesize \texttt{f94922074@ntu.edu.tw}
\and 
Mong-Jen Kao \\[6pt] \small Research Center for \\ \small IT Innovation \\ \small Academia Sinica \\ \small Taiwan \\[2pt] \footnotesize \texttt{mong@citi.sinica.edu.tw}
\and 
D.-T. Lee \\[6pt] \small Dep. of Computer Science \\ \small and Engineering \\ \small National Chung-Hsing Univ. \\ \small Taiwan \\[2pt] \footnotesize \texttt{dtlee@ieee.org}}

\date{}

\begin{document}

\maketitle


\begin{abstract}
We consider the problem of computing the time-convex hull of a point set under the general $L_p$ metric in the presence of a straight-line highway in the plane. The traveling speed along the highway is assumed to be faster than that off the highway, and the shortest time-path between a distant pair may involve traveling along the highway. 
The time-convex hull $\OP{TCH}(P)$ of a point set $P$ is the smallest set containing both $P$ and \emph{all} shortest time-paths between any two points in $\OP{TCH}(P)$.
In this paper we give an algorithm that computes the time-convex hull under the $L_p$ metric in optimal $\MC{O}\BIGP{n\log n}$ time for a given set of $n$ points and a real number $p$ with $1\le p \le \infty$.
\end{abstract}

%


\section{Introduction}

Path planning, in particular, shortest time-path planning, in complex transportation networks has become an important yet challenging issue in recent years. With the usage of heterogeneous moving speeds provided by different means of transportation, the \emph{time-distance} between two points, i.e., the amount of time it takes to go from one point to the other, is often more important than their straight-line distance.
With the reinterpretation of distances by the time-based concept, fundamental geometric problems such as convex hull, Voronoi diagrams, facility location, etc. have been reconsidered recently in depth and with insights~\cite{CHLiu12,Bae:2006:OCC:2172842.2172870,Abellanas:2003:VDS:846035.846044}.

%


%

From the theoretical point of view, straight-line highways which provide faster moving speed and which we can enter and exit at any point is one of the simplest transportation models to explore.
The speed at which one can move along the highway is assumed to be $v > 1$, while the speed off the highway is $1$.
Generalization of convex hulls in the presence of highways was introduced by Hurtado et al.~\cite{Hurtado99}, who suggested that the notion of convexity be defined by the inclusion of shortest time paths, instead of straight-line segments, i.e., a set $S$ is said to be {\bf convex} if it contains the shortest time-path between any two points of $S$. Using this new definition, the time-convex hull $\OP{TCH}(P)$ for a set $P$ is the closure of $P$ with respect to the inclusion of shortest time-paths.

%


%

In following work, Palop~\cite{Palop03} studied the structure of $\OP{TCH}(P)$ in the presence of a highway and showed that it is composed of convex clusters possibly together with segments of the highway connecting all the clusters.
A particularly interesting fact implied by the hull-structure is that, the shortest time-path between each pair of inter-cluster points must contain a piece of traversal along the highway, while similar assertions do not hold for intra-cluster pairs of points:
A distant pair of points $(p,q)$ whose shortest time-path contains a segment of the highway could still belong to the same cluster, for there may exist other points from the same cluster whose shortest time-path to either $p$ or $q$ does not use the highway at all.
This suggests that, the structure of $\OP{TCH}(P)$ in some sense 
indicates the degree of convenience provided by the underlying transportation network.
We are content with clusters of higher densities, i.e., any cluster with a large ratio between the number of points of $P$ it contains and the area of that cluster.
%
For sparse clusters, we may want to break 
them and benefit distant pairs they contain by enhancing the transportation infrastructure.

%


%

The approach suggested by Palop~\cite{Palop03} for the presence of a highway involves enumeration of shortest time-paths between all pairs of points and hence requires $\Theta\BIGP{n^2}$ time, where $n$ is the number of points.
%
This problem was later studied by Yu and Lee~\cite{Yu:2007:TCH:1270397.1271493}, who proposed an approach based on incremental point insertions in a highway-parallel monotonic order. However, the proposed algorithm does not return the correct hull in all circumstances as particular cases were overlooked.
The first sub-quadratic algorithm was given by Aloupis et al.~\cite{Aloupis:2010:HHR:1598101.1598676}, who proposed an $\MC{O}\BIGP{n\log^2n}$ algorithm for the $L_2$ metric and an $\MC{O}\BIGP{n\log n}$ algorithm for the $L_1$ metric, following the incremental approach suggested by~\cite{Yu:2007:TCH:1270397.1271493} with careful case analysis.
%
To the best of our knowledge, no previous results regarding metrics other than $L_1$ and $L_2$ were presented. 

%

\paragraph{Our Focus and Contribution.}

In this paper we address the problem of computing the time-convex hull of a point set in the presence of a straight-line highway under the $L_p$ metric for a given real number $p$ with $1\le p\le \infty$.
First, we adopt the concept of \emph{wavefront propagation}, a notion commonly used for path planning~\cite{CHLiu12,Aichholzer:2002:QPS:513400.513420}, and derive basic properties required for depicting the hull structure under the general $L_p$ metric.
When the shortest path between two points is not uniquely defined, e.g., in $L_1$ and $L_\infty$ metrics, we propose a re-evaluation on the existing definition of convexity.
Previous works concerning convex hulls under metrics other than $L_2$, e.g., Ottmann et al.~\cite{Ottmann1984157} and Aloupis et al.~\cite{Aloupis:2010:HHR:1598101.1598676}, assume a particular path to be taken when multiple choices are available.
However, this assumption allows the boundary of a convex set to contain reflex angles, which in some sense deviates from the intuition of a set being convex.

%


%

In this work we adopt the definition that requires a convex set to include \emph{every shortest path} between any two points it contains.
Although this definition fundamentally simplifies the shapes of convex sets for $L_1$ and $L_\infty$ metrics, we show that the nature of the problem is not altered when time-based concepts are considered. 
In particular, the problem of deciding whether any pair of the given points belong to the same cluster under the $L_p$ metric requires $\Omega\BIGP{n\log n}$ time under the algebraic computation model~\cite{Ben-Or:1983:LBA:800061.808735}, for all $1\le p\le \infty$.

%


%

Second, we provide an optimal $\MC{O}\BIGP{n\log n}$ algorithm for computing the time-convex hull for a given set of points.
The known algorithm due to Aloupis et al.~\cite{Aloupis:2010:HHR:1598101.1598676} stems from a scenario in the cluster-merging step where we have to check for the existence of intersections between a line segment and a set of convex curves composed of parabolae and line segments, which leads to their $\MC{O}\BIGP{n\log^2n}$ algorithm. 
In our paper, we tackle this situation by making an observation on the duality of cluster-merging conditions and reduce the problem to the geometric query of deciding if any of the given points lies above a line segment of an arbitrary slope.
This approach greatly simplifies the algorithm structure and can be easily generalized to other $L_p$-metrics for $1\le p\le \infty$.
For this particular geometric problem, we use a 
data structure due to Guibas et al~\cite{Guibas:1990:CIT:320176.320195} 
to answer this query in logarithmic time. 
All together this yields our $\MC{O}\BIGP{n\log n}$ algorithm.
We remark that, although 
our adopted definition of convexity
simplifies the shape of convex sets under the $L_1$ and the $L_\infty$ metrics, the algorithm we propose does not 
take 
advantage of this specific property and also works for the original notion for which only a particular path is to be included.

%


%



\section{Preliminaries}

In this section, we give precise definitions of the notions as well as sketches of previously known properties that are essential to present our work.
%
%
We begin with the general $L_p$ distance metric and basic time-based concepts.

\begin{definition}[Distance in the $L_p$-metrics]
For any real number $p \ge 1$ and any two points $q_i, q_j \in \MBB{R}^n$ with coordinates $\BIGP{i_1,i_2,\ldots,i_n}$ and $\BIGP{j_1,j_2,\ldots,j_n}$,
the distance between $q_i$ and $q_j$ under the $L_p$-metric is defined to be
$d_p\BIGP{q_i,q_j} = \BIGP{\sum_{k=1}^n \BIGC{i_{k}-j_{k}}^{p}}^{\frac{1}{p}}.$
\end{definition}
Note that when $p$ tends to infinity, $d_p(q_i,q_j)$ converges to $\max_{1\le k\le n}\BIGC{i_k-j_k}$.
This gives the definition of the distance function in the $L_\infty$-metric, which is $d_\infty(q_i,q_j) = \max_{1\le k\le n}\BIGC{i_k-j_k}$.
For the rest of this paper, we use the subscript $p$ to indicate the specific $L_p$-metric, and the subscript $p$ will be omitted when there is no ambiguity.

\smallskip

%

A \emph{transportation highway} $\MC{H}$ in $\MBB{R}^n$ is a hyperplane in which the moving speed in $\MC{H}$ is $v_\MC{H}$, where $1<v_\MC{H}\le \infty$, while the moving speed off $\MC{H}$ is assumed to be 
\emph{unit}.
%
Given the moving speed in the space, we can define the \emph{time-distance} between any two points in $\MBB{R}^n$.
%

\begin{definition}
For any $q_i,q_j \in \MBB{R}^n$, a continuous curve $\MC{C}$ connecting $q_i$ and $q_j$ is said to be a \emph{shortest time-path} if the traveling time required along $\MC{C}$ is minimum among all possible curves connecting $q_{i}$ and $q_{j}$.
The traveling time 
required along $\MC{C}$ is referred to as the time-distance between $q_i$ and $q_j$, denoted $\hat{d}(q_i,q_j)$.
\end{definition}

For any two points $q_i$ and $q_j$, 
let $\OP{STP}(q_i,q_j)$ denote the set of shortest time-paths between $q_i$ and $q_j$.
For any $\MC{C} \in \OP{STP}(q_i,q_j)$, we say that $\MC{C}$ \emph{enters the highway $\MC{H}$} if 
$\MC{C} \cap \MC{H} \ne \emptyset$.
%
The \emph{walking-region} of a point $q \in \MBB{R}^n$, denoted $\OP{WR}(q)$, is defined to be the set of points whose set of shortest time-paths to $q$ contains a time-path that does not enter the highway $\MC{H}$.
For any $\MC{C} \in \OP{STP}(q_i,q_j)$, we say that $\MC{C}$ \emph{uses the highway $\MC{H}$} if 
$\MC{C} \cap \MC{H}$ contains a piece with non-zero length, i.e., at some point $\MC{C}$ enters the highway $\MC{H}$ and walks along it.

\paragraph{Convexity and time-convex hulls.}
In classical definitions, a set of points is said to be \emph{convex} if it contains every line segment joining each pair of points in the set, and the convex hull of a set of points $Q \subseteq \MBB{R}^n$ is 
the minimal convex set containing $Q$.
%
%
%
When time-distance is considered, the concept of convexity as well as convex hulls with respect to time-paths 
is defined analogously.
A set of points is said to be \emph{convex with respect to time}, or, \emph{time-convex}, if it contains every shortest time-path joining each pair of points in the set.

\begin{definition}[Time-convex hull]
The time-convex hull, of a set of points $Q \subseteq \MBB{R}^n$, denoted $\OP{TCH}(Q)$, is the minimal time-convex set containing $Q$.
\end{definition}

Although the aforementioned concepts 
are defined in $\MBB{R}^n$ space, 
in this paper we work in $\MBB{R}^2$ plane with an axis-parallel highway placed on the $x$-axis 
as higher dimensional space does not give further insights:
When considering the shortest time-paths between two points in higher dimensional space, it suffices to consider the specific plane that is orthogonal to $\MC{H}$ and that contains the two points.
%
%


\paragraph{Time-convex hull under the $L_1$ and the $L_2$ metrics.}

The structure of time-convex hulls under the $L_1$ and the $L_2$ metrics has been studied in a series of work~\cite{Aloupis:2010:HHR:1598101.1598676,Palop03,Yu:2007:TCH:1270397.1271493}. Below we review important properties.
See also Fig.~\ref{fig-l1-l2-walking-region} for an illustration.

\begin{figure*}[b]
{\includegraphics[scale=0.8]{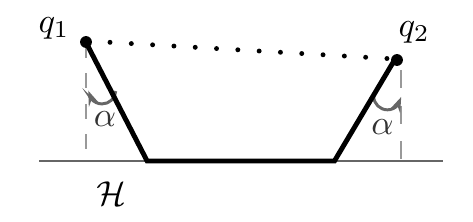}}
{\includegraphics[scale=0.8]{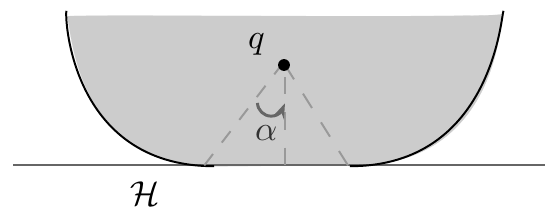}}
{\includegraphics[scale=0.8]{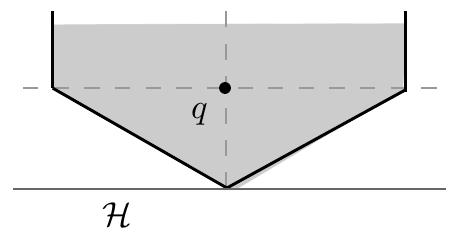}}
\caption{(a) The only two possible paths for being a shortest time-path between $q_1$ and $q_2$ in $L_2$. (b)(c) The walking regions of a point $q \in \MBB{R}^2$ under $L_2$ and $L_1$, respectively. 
}
\label{fig-l1-l2-walking-region}
\end{figure*}

\begin{prop}[\cite{Palop03,Yu:2007:TCH:1270397.1271493}] \label{prop-l2-wr}
For the $L_2$-metric and any point $q = (x_q,y_q)$ with $y_q \ge 0$, we have the following properties.
\begin{compactenum}
	\item
		If a shortest time-path starting from $q$ uses the highway $\MC{H}$, then it must enter the highway with an incidence angle $\alpha = \arcsin 1/v_\MC{H}$ toward the direction of the destination.
		

	\item
		The walking region of a point $q \in \MBB{R}^2$ is characterized by the following two parabolae:
				(a) \emph{right discriminating parabola}, which is the curve satisfying 
				\begin{align*}
				\begin{cases}
					x\ge x_q + y_q\tan\alpha, \quad \text{and} \\
					\sqrt{\BIGP{x-x_q}^2+\BIGP{y-y_q}^2} = y_q\sec\alpha + y\sec\alpha 
					+ \frac{1}{v_\MC{H}}\BIGP{\BIGP{x-y\tan\alpha}-\BIGP{x_q+y_q\tan\alpha}}.
				\end{cases}
				\end{align*}
%
		%
		(b) The \emph{left discriminating parabola} 
		is symmetric to the right discriminating parabola with respect to the line $x=x_q$. 
\end{compactenum}
\end{prop}

\begin{prop}[\cite{Aloupis:2010:HHR:1598101.1598676}] \label{prop-l1-previous}
For the $L_1$-metric, the walking region of a point $q = (x_q,y_q)$ with $y_q \ge 0$ is formed by 
the intersection of the following regions: (a) the vertical strip $x_q - {y_q}/{\beta} \le x \le x_q - y_q / \beta$, and (b) $y \ge \pm \beta\BIGP{x-x_q}$, where $\beta = \frac{1}{2}\BIGP{1-\frac{1}{v_\MC{H}}}$.
\end{prop}

%



\section{Hull-Structure under the General $L_p$-Metrics}

In this section, we derive necessary properties to describe the structure of time-convex hulls under the general $L_p$ metrics. First, we adopt the notion of \emph{wavefront propagation}~\cite{CHLiu12,Aichholzer:2002:QPS:513400.513420}, which is a well-established model used in path planning, and derive the behavior of a shortest time-path between any two points. 
Then we show how the corresponding walking regions are formed, followed by a description of the desired structural properties.
%
%

\paragraph{Wavefronts and Shortest Time-Paths.}

For any $q \in \MBB{R}^2$, $t \ge 0$, and $p\ge 1$, the wavefront with \emph{source} $q$ and \emph{radius} $t$ under the $L_p$-metric is defined as 
$$W_p(q,t) = \BIGBP{s\colon s\in\MBB{R}^2, \hat{d}_p(q,s) = t}.$$
Literally, $W_p(q,t)$ is the set of points whose time-distances to $q$ are exactly $t$.
Fig.~\ref{fig-p-circles}~(a) shows the wavefronts, i.e., the ``unit-circles'' under the $L_p$ metric, or, the $p$-circles, for different $p$ with $0<p\le \infty$ when the highway is not used.
The shortest time-path between $q$ and any point $q^\prime \in W_p(q,t)$ is the trace on which $q^\prime$ moves as $t$ changes smoothly to zero,
which is a straight-line joining $q$ and $q^\prime$.

When the highway $\MC{H}$ is present and the time-distance changes, deriving the behavior of a shortest time-path that uses $\MC{H}$ becomes tricky.
Let $q_1, q_2 \in \MBB{R}^2$, $q_1 \ne q_2$, be two points in the plane, and let $\hat{t}_{\nicefrac{1}{2}}(q_1,q_2) \ge 0$ be the smallest real number such that 
$$\OP{Bisect}(q_1,q_2) \equiv W_p\BIGP{q_1,\hat{t}_{\nicefrac{1}{2}}(q_1,q_2)} \cap W_p\BIGP{q_2,\hat{t}_{\nicefrac{1}{2}}(q_1,q_2)} \ne \emptyset.$$
In other words, $\OP{Bisect}(q_1,q_2)$ is the set of points at which 
$W_p(q_1,t)$ and $W_p(q_2,t)$ meet for the first time. The following lemma shows that $\OP{Bisect}(q_1,q_2)$ characterizes the set of ``\emph{middle points}'' of all shortest time-paths between $q_1$ and $q_2$.

\begin{lemma} \label{lemma-shortest-time-path-bisecting-set}
For each $\MC{C} \in \OP{STP}(q_1,q_2)$, we have $\MC{C} \cap \OP{Bisect}(q_1,q_2) \ne \emptyset$.
Moreover, for each $q \in \OP{Bisect}(q_1,q_2)$, there exists $\MC{C}^\prime \in \OP{STP}(q_1,q_2)$ such that $q \in \MC{C}^\prime$.
\end{lemma}

Given the set $\OP{Bisect}(q_1,q_2)$, a shortest time-path between $q_1$ and $q_2$ can be obtained by joining $\MC{C}_1$ and $\MC{C}_2$, where $\MC{C}_1 \in \OP{STP}(q_1,q)$ and $\MC{C}_2 \in \OP{STP}(q,q_2)$ for some $q \in \OP{Bisect}(q_1,q_2)$.
%
By expanding the process in a recursive manner 
we get a set of middle points.
Although the cardinality of the set we identified is countable while any continuous curve in the plane contains uncountably infinite points, it is not difficult to see that, the set of points we locate is \emph{dense}\footnote{
\emph{Dense} is a concept used in classical analysis 
to indicate that any element of one set can be approximated to any degree by elements of a subset being dense within.
}
in the underlying curve, and therefore can serve as a representative.

\smallskip

To describe the shape of a wavefront when the highway may be used, 
we need the following lemma regarding the propagation of wavefronts.

\begin{lemma} \label{lemma-wavefront-recursion}
Let $\MC{C}^\circ_p(q,t)$ denote the $p$-circle with center $q$ and radius $t$.
Then $W_p(q,t)$ is formed by the boundary of
$$\MC{C}^\circ_p(q,t) \cup \bigcup_{s\colon s\in \MBB{R}^2, \\ \hat{d}_p(q,s) < d_p(q,s) \le t} W_p\BIGP{s,t-\hat{d}_p(q,s)}.$$
\end{lemma}

{
In the following we discuss the case when $1<p<\infty$ and leave the discussion of shortest time-paths in $L_\infty$ to the appendix for further reference.}
Let $\MC{H}$ be the highway placed on the x-axis with moving speed $v_\MC{H} > 1$.
%
For any $t \ge 0$, any $0\le x \le t$, and any $1<p<\infty$, we use $y_p(x,t) = \BIGP{t^p-\BIGC{x}^p}^{1/p}$ to denote the y-coordinate of the specific point on the p-circle with x-coordinate $x$.
Let $\hat{x}_t = t\cdot v_\MC{H}^{1/(1-p)}$.
We have the following lemma regarding $W_p(q,t)$. Also refer to Fig.~\ref{fig-p-circles}~(b) for an illustration.

\begin{figure*}[t]
\hspace{-12pt}\begin{minipage}{0.22\textwidth}
\includegraphics[scale=1.1]{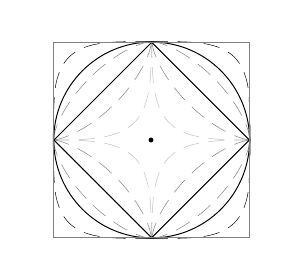}
\end{minipage}
\begin{minipage}{0.42\textwidth}
\includegraphics[scale=1.1]{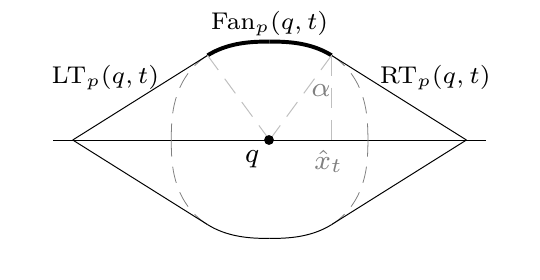}
\end{minipage}
\begin{minipage}{0.35\textwidth}
\includegraphics[scale=1.1]{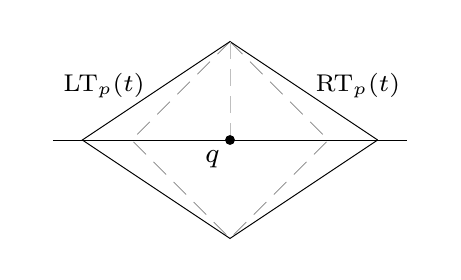}
\end{minipage}
\caption{(a) $p$-circles for different values of $p$: bold rhombus for $p=1$, bold circle for $p=2$, and bold square for $p=\infty$. (b) $W_p(q,t)$ for a point $q \in \MC{H}$, $v_\MC{H}<\infty$, and $p>1$, where the angle $\alpha$ satisfies $\sin\alpha = \frac{\hat{x}_t}{ \sqrt{\hat{x}_t^2+y_p(\hat{x}_t)^2} }$.
(c) $W_p(q,t)$ for $v_\MC{H}<\infty$ and $p = 1$.}
\label{fig-p-circles}
\end{figure*}

\begin{lemma} \label{lemma-wavefront-from-highway}
For $1<p<\infty$, $v_\MC{H} < \infty$, and a point $q \in \MC{H}$ which we assume to be $(0,0)$ for the ease of presentation, the upper-part of $W_p(q,t)$ that lies above $\MC{H}$ consists of the following three pieces:
\begin{compactitem}
	\item
		$\OP{Fan}_p(q,t)$: the circular-sector of the $p$-circle with radius $t$, ranging from $\BIGP{-\hat{x}_t,y_p(-\hat{x}_t,t)}$ to $\BIGP{\hat{x}_t,y_p(\hat{x}_t,t)}$.
	\item
		$\OP{LT}_p(q,t)$, $\OP{RT}_p(q,t)$: two line segments joining $\BIGP{-v_\MC{H}\cdot t,0}$, $\BIGP{-\hat{x}_t,y_p(-\hat{x}_t,t)}$, and $\BIGP{\hat{x}_t,y_p(\hat{x}_t,t)}$, $\BIGP{v_\MC{H}\cdot t,0}$, respectively. Moreover, $\OP{LT}_p(q,t)$ and $\OP{RT}_p(q,t)$ are tangent to $\OP{Fan}_p(q,t)$.
\end{compactitem}
The lower-part that lies below $\MC{H}$ follows symmetrically.
For $v_\MC{H} = \infty$, the upper-part of $W_p(q,t)$ consists of a horizontal line $y = t$.
\end{lemma}

\smallskip

For each $1\le p < \infty$ and $1<v_\MC{H}\le \infty$, we define the real number $\alpha\BIGP{p,v_\MC{H}}$ as follows.
If $p = 1$ or $v_\MC{H} = \infty$, then $\alpha\BIGP{p,v_\MC{H}}$ is defined to be zero. Otherwise, $\alpha\BIGP{p,v_\MC{H}}$ is defined to be 
$$\arcsin \frac{v_\MC{H}^{1/(1-p)}}{\sqrt{v_\MC{H}^{2/(1-p)}+\BIGP{1-v_\MC{H}^{p/(1-p)}}^{2/p}}}.$$
Note that, when $p = 2$, this is exactly $\arcsin(1/v_\MC{H})$.
%
For brevity, we simply use $\alpha$ when there is no ambiguity.
The behavior of a shortest time-path that takes the advantage of traversal along the highway is characterized by the following lemma.

\begin{lemma} \label{lemma-incidence-angle}
For any point $q = \BIGP{x_q,y_q}$, $1\le p < \infty$, and $1<v_\MC{H}\le\infty$, if a shortest time-path starting from $q$ uses the highway $\MC{H}$, then it must enter the highway with an incidence angle $\alpha$.
\end{lemma}



\paragraph{Walking Regions.}

%
%
For any point $q = (x_q,y_q) \in \MBB{R}^2$ with $y_q \ge 0$, let $q^+_\MC{H}$ and $q^-_\MC{H}$ be two points 
located at $\BIGP{x_q \pm y_q\tan\alpha,0}$, respectively. 
By Lemma~\ref{lemma-incidence-angle}, we know that, $q^+_\MC{H}$ and $q^-_\MC{H}$ are exactly the points at which any shortest time-path from $q$ will enter the highway if needed.
This gives the walking region for any point.
{Let $\alpha = {\pi}/{4}$ when $p = \infty$.} 
The following lemma is an updated version of Proposition~\ref{prop-l2-wr} for general $p$ with $1\le p\le \infty$.

\begin{lemma} \label{lemma-wr-general}
For any $p$ with $1\le p\le \infty$ and any point $q = (x_q,y_q)$ with $y_q \ge 0$, 
$\OP{WR}_p(q)$ is characterized by the following two curves:
(a) \emph{right discriminating curve}, which is the curve $q^\prime = (x^\prime,y^\prime)$ satisfying $x^\prime\ge x_q + y_q\tan\alpha$ and
\begin{align*}
\BIGC{\overline{qq^\prime}}_p = \BIGC{\overline{qq^+_\MC{H}}}_p + \BIGC{\overline{q^\prime q^{\prime-}_\MC{H}}}_p + \frac{1}{v_\MC{H}}\BIGC{\overline{q^+_\MC{H}q^{\prime-}_\MC{H}}}_p.
\end{align*}
(b) The \emph{left discriminating curve} is symmetric with respect to the line $x=x_q$. 
\end{lemma}

%

%

For any point $q$, let $\OP{WR}_{\ell}(q)$ and $\OP{WR}_{r}(q)$ denote the left- and right- discriminating curves of $\OP{WR}(q)$, respectively.
%
%
%
We have the following \emph{dominance property} of the walking regions.

\begin{lemma} \label{lemma-wr-dominance}
Let $q_1 = (x_1,y_1)$ and $q_2 = (x_2, y_2)$ be two points such that $x_1 \le x_2$. If $y_1 \ge y_2$, then $\OP{WR}_{\ell}(q_2)$ lies to the right of $\OP{WR}_{\ell}(q_1)$.
Similarly, if $y_1 \le y_2$, then $\OP{WR}_{r}(q_1)$ lies to the left of $\OP{WR}_{r}(q_2)$.
\end{lemma}

\begin{figure*}[htp]
\centering\fbox{\includegraphics[scale=0.85]{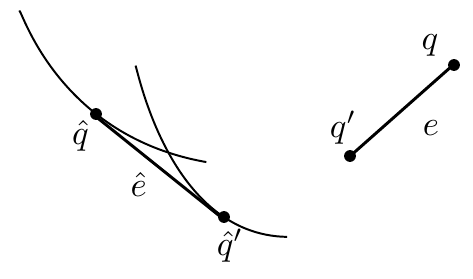}}
\caption{\small The left side shows the left boundary of the walking region for the edge $e=(q^\prime,q)$ on the right side.}
\label{fig-wr-edge}
\end{figure*}
Lemma~\ref{lemma-wr-dominance} suggests that, to describe the leftmost and the rightmost boundaries of the walking-regions for a set of points, it suffices to consider the \emph{extreme points}.
Let $e = \overline{q_1q_2}$ be a line segment between two points $q_1$ and $q_2$, where $q_1$ lies to the left of $q_2$.
If $e$ has non-positive slope, then the left-boundary of the walking region for $e$ is dominated by $\OP{WR}_\ell(q_1)$.
Otherwise, we have to consider $\bigcup_{q \in e}\OP{WR}_\ell(q)$.
By parameterizing each point of $e$, it is not difficult to see that the left-boundary consists of $\OP{WR}_\ell(q_1)$, $\OP{WR}_\ell(q_2)$, and their common tangent line.
%
See also Fig.~\ref{fig-wr-edge} for an illustration.

%




\paragraph{Closure and Time-Convex Hull of a Point Set.}

By Lemma~\ref{lemma-shortest-time-path-bisecting-set}, to obtain the union of possible shortest time-paths, it suffices to consider the set of all possible bisecting sets that arise inside the recursion.
We begin with the closure between pairs of points.

\begin{lemma} \label{lemma-closure-pair-of-points}
Let $q_1, q_2 \in \MBB{R}^2$ be two points. When the highway is not used, the set of all shortest time-paths between $q_1$ and $q_2$ is:
\begin{compactitem}
	\item
		The smallest bounding rectangle of 
		$\BIGBP{q_1,q_2}$, when $p = 1$.
	
	\item
		The straight line segment $\overline{q_1q_2}$ joining $q_1$ and $q_2$, when $1<p<\infty$.
		
	\item
		The smallest bounding parallelogram whose slopes of the four sides are $\pm 1$, i.e., a rectangle rotated by $45^\circ$, that contains $q_1$ and $q_2$. 
\end{compactitem}
\end{lemma}

Lemma~\ref{lemma-closure-pair-of-points} suggests that when the highway is not used and when $1<p<\infty$, the closure, or, convex hull, of a point set $\MB{S}$ with respect to the $L_p$-metric is identical to that in $L_2$, while in $L_1$ and $L_\infty$ the convex hulls are given by the bounding rectangles and bounding square-parallelograms.
%

%

\smallskip

\begin{wrapfigure}[9]{l}{0.52\textwidth}
\fbox{\includegraphics[scale=0.7]{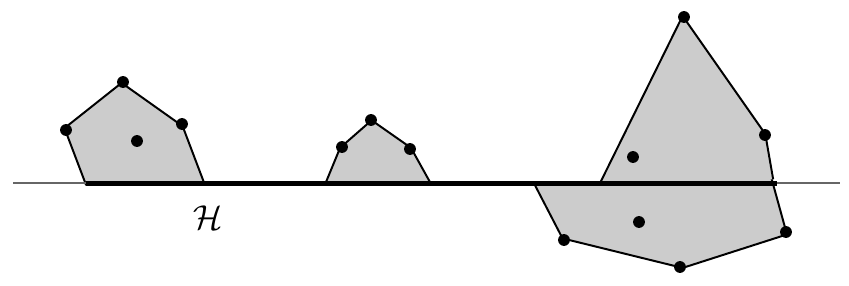}}
\caption{\small Time-convex hull for a set of points under the $L_p$ metric where $1<p<\infty$.}
\label{fig-time-convex-hull-l2}
\end{wrapfigure}
\noindent
When the highway may be used, the structure of the time-convex hull under the general $L_p$-metric 
consists of a set of clusters arranged in a way such that the following holds:
(1) Any shortest time-path between intra-cluster pair of points must 
use the highway.
(2) If any shortest time-path between two points does not use the highway, then the two points must belong to the same cluster.
%
%
Fig.~\ref{fig-time-convex-hull-l2} and Fig.~\ref{fig-time-convex-hull-linfty} illustrate examples of the time-convex hull for the $L_p$ metrics with $1<p<\infty$ and $p=\infty$, respectively.
Note that, the shape of the closure for each cluster does depend on $p$ and $v_\MC{H}$, as they determine the incidence angle $\alpha$.
%

%

\begin{figure*}[t]
\hspace{-10pt}
{\includegraphics[scale=0.7]{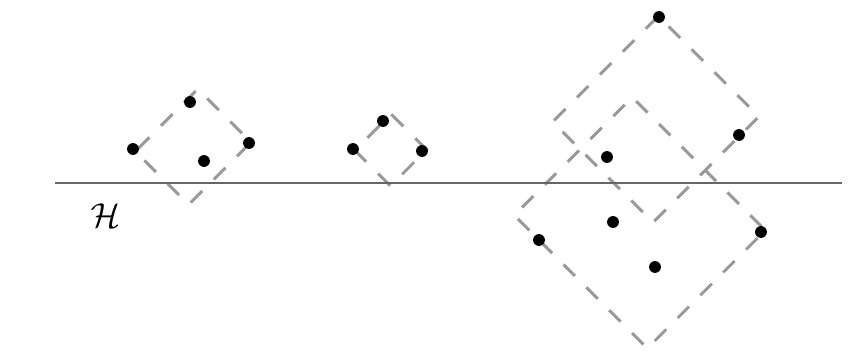}}
{\includegraphics[scale=0.7]{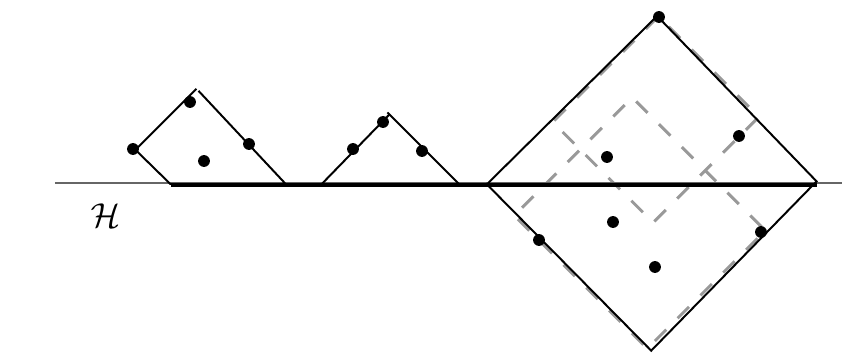}}
\caption{(a) The closure of each cluster under the $L_\infty$ metric when the highway $\MC{H}$ is not considered. (b) The closure,
i.e., the time-convex hull, for the $L_\infty$ metric.
}
\label{fig-time-convex-hull-linfty}
\end{figure*}



%





\section{Constructing the Time-Convex Hull}

In this section, we present our algorithmic results for this problem. First, we show that, although our definition of convexity simplifies the structures of the resulting convex hulls, e.g., in $L_1$ and $L_\infty$, the problem of deciding if any given pair of points belongs to the same cluster already requires $\Omega\BIGP{n\log n}$ time. 
Then we present our optimal $\MC{O}\BIGP{n\log n}$ algorithm.

\subsection{Problem Complexity}

We make a reduction from the \emph{minimum gap problem}, which is a classical problem known to have the problem complexity of $\Theta\BIGP{n\log n}$.
Given $n$ real numbers $a_1, a_2, \ldots, a_n$ and a target gap $\epsilon > 0$, the minimum gap problem is to decide if there exist some $i,j$, $1\le i,j\le n$, such that $\BIGC{a_i-a_j} \le \epsilon$.

\smallskip

For any $y \ge 0$, consider the point $q(y) = (0,y)$.
Let $\MC{C}_{q(y)}\colon \BIGP{x,f_{q(y)}(x)}$ denote the right discriminating curve of $\OP{WR}_p\BIGP{q(y)}$.
For any $\epsilon > 0$, let $y_0(p,\epsilon)$ denote the specific real number such that $f_{q\BIGP{y_0(p,\epsilon)}}\BIGP{\epsilon} = y_0(p,\epsilon)$.
Our reduction is done as follows.
Given a real number $p \ge 1$ and an instance $\MC{I}$ of minimum gap, we create a set $\MB{S}$ consisting of $n$ points $q_1, q_2, \ldots, q_n$, where $q_i = \BIGP{a_i,y_0\BIGP{p,\epsilon}}$ for $1\le i\le n$. 
The following lemma shows the correctness of this reduction and establishes the $\Omega\BIGP{n\log n}$ lower bound.

\begin{lemma} \label{lemma-lower-bound}
$y_0\BIGP{p,\epsilon}$ is well-defined for all $p$ with $1\le p\le \infty$ and all $\epsilon > 0$. Furthermore, the answer to the minimum gap problem on $\MC{I}$ is ``yes'' if and only if the number of clusters in the time-convex hull of $\MB{S}$ is less than $n$.
\end{lemma}

\begin{corollary}
Given a set of points $\MB{S}$ in the plane, a real number $p$ with $1\le p\le \infty$, and a highway $\MC{H}$ placed on the x-axis, the problem of deciding if any given pair of points belongs to the same cluster requires $\Omega\BIGP{n\log n}$ time.
\end{corollary}

\subsection{An Optimal Algorithm}

In this section, we present our algorithm for constructing the time-convex hull for a given point set $\MB{S}$ under a given metric $L_p$ with $p\ge 1$.
The main approach is to insert the points incrementally into the partially-constructed clusters in ascending order of their x-coordinates.
%
In order to prevent a situation that leads to an undesirably complicated query encountered in the previous work by Aloupis et al.~\cite{Aloupis:2010:HHR:1598101.1598676}, we exploit the symmetric property of cluster-merging conditions and reduce the sub-problem to the following geometric query.
\begin{definition}[One-Sided Segment Sweeping Query]
Given a set of points $\MB{S}$ in the plane, for any line segment $\OP{L}$ of finite slope, the \emph{one-sided segment sweeping query}, denoted $\MC{Q}(\OP{L})$, asks if $\MB{S} \cap \OP{L}^+$ is empty, where $\OP{L}^+$ is the intersection of the half-plane above $\OP{L}$ and the vertical strip defined by the end-points of $\OP{L}$.
That is,
we ask if there exists any point $p \in \MB{S}$ such that $p$ lies above 
$\OP{L}$.
\end{definition}
In the following, we first describe the algorithm and our idea in more detail, assuming the one-sided segment sweeping query is available.
Then we show how this query can be answered efficiently.

\smallskip

%

The given set $\MB{S}$ of points is partitioned into two subsets, one containing those points lying above $\MC{H}$ and the other containing the remaining. We compute the time-convex hull for the two subsets separately, followed by using a linear scan on the clusters created on both sides to obtain the closure for the entire point set. Below we describe how the time-convex hull for each of the two subsets can be computed.

Let $q_1, q_2, \ldots, q_n$ be the set of points sorted in ascending order
of their $x$-coordinates with ties broken by their $y$-coordinates.
During the execution of the algorithm, we maintain the set of clusters the algorithm has created so far, which we further denote by $\MC{C} = \BIGBP{C_1, C_2, \ldots, C_k}$.
For ease of presentation, we denote the left- and right-boundary of the walking region of $C_i$ by $\OP{WR}_\ell(C_i)$ and $\OP{WR}_r(C_i)$, respectively.
Furthermore, we use $q \in \OP{WR}_\ell(C_i)$ or $q \in \OP{WR}_r(C_i)$ to indicate that point $q$ lies to the right of $\OP{WR}_\ell(C_i)$ or to the left of $\OP{WR}_r(C_i)$, respectively.

\smallskip

In iteration $i$, $1\le i\le n$, the algorithm inserts $q_i$ into $\MC{C}$ and checks if a new cluster has to be created or if existing clusters have to be merged.
This is done in the following two steps.
\begin{compactenum}[(a)]
	\item
		\emph{Point inclusion test}.
		In this step, we check if there exists any $j$, $1\le j\le k$, such that $q_i \in \OP{WR}_r(C_j)$. If not, then a new cluster $C_{k+1}$ consisting of the point $q_i$ is created and we enter the next iteration.
		Otherwise, the smallest index $j$ such that $q_i \in \OP{WR}_r(C_j)$ is located.
		The clusters $C_j, C_{j+1}, \ldots, C_k$ and the point $q_i$ are merged into 
		one cluster, which will in turn replace $C_j, C_{j+1}, \ldots, C_k$.
		Let $\MB{E}$ be the set of newly created edges on the upper-hull of this cluster whose slopes are positive. 
		Then we proceed to step (b).
		
	\item
		\emph{Edge inclusion test}.
		Let $k$ be the number of clusters, and $x_0$ be the $x$-coordinate of the leftmost point in $C_k$.
		Pick an arbitrary edge $e \in \MB{E}$, let $\hat{e}$ denote the line segment appeared on $\OP{WR}_\ell(e)$ to which $e$ corresponds, and let $\hat{e}(x_0)$ be the intersection of $\hat{e}$ with the half-plane $x\le x_0$.
		Then we invoke the one-sided segment sweeping query $\MC{Q}\BIGP{\hat{e}(x_0)}$.
		%
		If no point lies above $\hat{e}(x_0)$, then $e$ is removed from $\MB{E}$.
		Otherwise, $C_k$ is merged with $C_{k-1}$. Let $e^\prime$ be the newly created bridge edge between $C_k$ and $C_{k-1}$. If $e^\prime$ has positive slope, then it is added to the set $\MB{E}$.
		This procedure is repeated until the set $\MB{E}$ becomes empty.
		
\end{compactenum}

\medskip

%


An approach has been proposed to resolve the \emph{point inclusion test}
efficiently, e.g., Yu and Lee~\cite{Yu:2007:TCH:1270397.1271493}, and Aloupis et al.~\cite{Aloupis:2010:HHR:1598101.1598676}.
%
Below we state the lemma directly and leave the technical details to the appendix for further reference.

\begin{lemma}[\cite{Yu:2007:TCH:1270397.1271493,Aloupis:2010:HHR:1598101.1598676}] \label{lemma-pt-inclusion}
For each iteration, say $i$, the smallest index $j$, $1\le j\le k$, such that $q_i \in \OP{WR}_r(C_j)$ can be located in amortized constant time.
\end{lemma}

To see that our algorithm gives the correct 
clustering, it suffices to argue the following two conditions: (1) Each cluster-merge our algorithm performs is valid. (2) At the end of each iteration, no more clusters have to be merged.

Apparently these conditions hold at the end of the first iteration, when $q_1$ is processed.
For each of the succeeding iterations, say $i$, if no clusters are merged in step (a), then the conditions hold trivially. Otherwise, the validity of the cluster-merging operations 
is guaranteed by Lemma~\ref{lemma-pt-inclusion} and the fact that if any point lies above $\hat{e}$, then it belongs to the walking region of $e$, meaning that the last cluster, $C_k$, has to be merged again. See also Fig.~\ref{fig-segment-query-reduction} for an illustration.

\smallskip

\begin{wrapfigure}[13]{l}{0.6\textwidth}
\centering\fbox{\includegraphics[scale=.9]{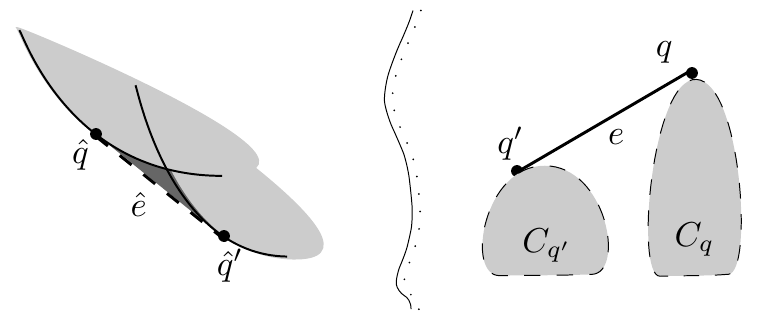}}
\caption{When two clusters $C_{q^\prime}$ and $C_q$ are merged and new hull edge $e$ is created, it suffices to check the new walking region $e$ corresponds to, i.e., the dark-gray area.}
\label{fig-segment-query-reduction}
\end{wrapfigure}
\noindent To see that the second condition holds,
let $e=(q^\prime, q) \in \MB{E}$ be a newly created hull edge, and let $C_{q^\prime}$ and $C_q$ be the two corresponding clusters that were merged.
%
%
By our assumption that the clusters are correctly created before $q_i$ arrives, we know that the walking-regions of $C_{q^\prime}$ and $C_q$ contain only points that do belong to them, i.e., the light-gray area in the left-hand side of Fig.~\ref{fig-segment-query-reduction} contains only points from $C_{q^\prime}$ or $C_q$.
Therefore, when $C_{q^\prime}$ and $C_q$ are merged and $e$ is created, it suffices to check for the existence of points other than $C_k$ inside the new walking region $e$ corresponds to, which is exactly the dark-gray area in Fig.~\ref{fig-segment-query-reduction}.
Furthermore, by the dominance property stated in Lemma~\ref{lemma-wr-dominance}, it suffices to check those edges with positive slopes.
This shows that at the end of each iteration when $\MB{E}$ becomes empty, no more clusters need to be merged.
We have the following theorem.
%

\begin{theorem} \label{thm-opt-algo}
Provided that the one-sided segment sweeping query can be answered in $Q(n)$ time using $P(n)$ preprocessing time and $S(n)$ storage, the time-convex hull for a given set $\MB{S}$ of $n$ points under the given $L_p$-metric can be computed in $\MC{O}\BIGP{n\log n + nQ(n) + P(n)}$ time using $\MC{O}\BIGP{n+S(n)}$ space.
\end{theorem}




\subsubsection{Regarding the One-Sided Segment Sweeping Query.}

%

Below we sketch how this query can be answered efficiently in logarithmic time.
Let $\MB{S}$ be the set of points, $\OP{L}$ be the line segment of interest, and $\MC{I}_\OP{L}$ be the vertical strip defined by the two end-points of $\OP{L}$.
We have the following observation, which relates the query $\MC{Q}(\OP{L})$ to the problem of computing the upper-hull of $\MB{S} \cap \MC{I}_\OP{L}$.

\begin{lemma} \label{lemma-segment-query-pruning}
Let $\MC{I}$ be an interval, $\MC{C}\colon \MC{I} \rightarrow \MBB{R}$ be a convex function, i.e., we have $\MC{C}\BIGP{\frac{1}{2}\BIGP{x_1+x_2}} \ge \frac{1}{2}\BIGP{\MC{C}(x_1)+\MC{C}(x_2)} \enskip \forall x_1,x_2 \in \MC{I}$, that is differentiable almost everywhere,
$\OP{L}$ be a segment with slope $\theta_L$, $-\infty < \theta_L < \infty$, and $q=\BIGP{x_q,\MC{C}(x_q)}$ be a point on the curve $\MC{C}$ such that 
$$\lim_{x\rightarrow x_q^-}\frac{\OP{d}\MC{C}(x)}{\OP{d}x} \ge \theta_L \ge \lim_{x\rightarrow x_q^+}\frac{\OP{d}\MC{C}(x)}{\OP{d}x}.$$
If $q$ lies under $\overleftrightarrow{\OP{L}}$, then the curve $\MC{C}$ never intersects $\OP{L}$.
\end{lemma}

To help compute the upper-hull of $\MB{S} \cap \MC{I}_\OP{L}$, we use a data structure due to Guibas et al~\cite{Guibas:1990:CIT:320176.320195}.
%
For a given simple path $\MC{P}$ of $n$ points with an $x$-sorted ordering of the points, with $\MC{O}(n\log\log n)$ preprocessing time and space, the upper-hull of any subpath $p \in \MC{P}$ can be assembled efficiently in $\MC{O}(\log n)$ time, represented by a balanced search tree that allows binary search on the hull edges.
Note that, $q_1, q_2, \ldots, q_n$ is exactly a simple path by definition.
The subpath to which $\MC{I}_\OP{L}$ corresponds can be located in $\MC{O}(\log n)$ time.
%
In $\MC{O}(\log n)$ time we can obtain the corresponding upper-hull and test the
%
condition specified in Lemma~\ref{lemma-segment-query-pruning}.
We conclude with the following lemma.

\begin{lemma} \label{thm-pe-query}
The \emph{one-sided segment sweeping query} can be answered in $\MC{O}(\log n)$ time, where $n$ is the number of points, using $\MC{O}(n\log n)$ preprocessing time and $\MC{O}(n\log\log n)$ space.
\end{lemma}




\section{Conclusion}

We conclude with a brief discussion as well as an overview on future work.
In this paper, we give an optimal algorithm for the time-convex hull in the presence of a straight-line highway under the general $L_p$-metric where $1\le p\le \infty$.
The structural properties we provide 
involve non-trivial geometric arguments.
We believe that our algorithm and the approach we use can serve as a base to the scenarios for which we have a more complicated transportation infrastructure, e.g., modern city-metros represented by line-segments of different moving speeds.
Furthermore, we believe that approaches supporting dynamic settings to a certain degree, e.g., point insertions/deletions, or, dynamic speed transitions, are also a nice direction to explore.



\bibliographystyle{plain}
\bibliography{highway-hull-lp-metric}

%
%
%


\newpage


\begin{appendix}

\normalsize

%
%


\section{Hull-Structure in General the $L_p$-Metrics}

\subsection{Wavefronts and Shortest Time-Paths}


\begin{ap_lemma}{\ref{lemma-shortest-time-path-bisecting-set}}
For each $\MC{C} \in \OP{STP}(q_1,q_2)$, we have $\MC{C} \cap \OP{Bisect}(q_1,q_2) \ne \emptyset$.
Moreover, for each $q \in \OP{Bisect}(q_1,q_2)$, there exists $\MC{C}^\prime \in \OP{STP}(q_1,q_2)$ such that $q \in \MC{C}^\prime$.
\end{ap_lemma}

\begin{proof}[Proof of Lemma~\ref{lemma-shortest-time-path-bisecting-set}]
For any $\MC{C} \in \OP{STP}(q_1,q_2)$ and any $q \in \MC{C}$, $\hat{d}\BIGP{q_1,q}$ is a continuous function taking the values from $0$ to $\hat{d}\BIGP{q_1,q_2}$. Therefore there exists a point $q_{\nicefrac{1}{2}}$ such that $\hat{d}\BIGP{q_1,q_{\nicefrac{1}{2}}} = \frac{1}{2}\hat{d}\BIGP{q_1,q_2}$. This implies $\hat{d}\BIGP{q_{\nicefrac{1}{2}},q_2} \le \frac{1}{2}\hat{d}\BIGP{q_1,q_2}$, where the equality must hold since $\MC{C}$ is a shortest time-path between $q_1$ and $q_2$.
This shows that $q_{\nicefrac{1}{2}} \in W\BIGP{q_1,\frac{1}{2}\hat{d}\BIGP{q_1,q_2}} \cap W\BIGP{q_2,\frac{1}{2}\hat{d}\BIGP{q_1,q_2}}$.
Furthermore, we have $W\BIGP{q_1,t} \cap W\BIGP{q_2,t} = \emptyset$ for all $t < \frac{1}{2}\hat{d}\BIGP{q_1,q_2}$, for otherwise it will result in a shorter time-path than $\MC{C}$, which is a contradiction. Therefore $q_{\nicefrac{1}{2}} \in \OP{Bisect}(q_1,q_2)$.
This also shows that $\hat{t}_{\nicefrac{1}{2}}(q_1,q_2) = \frac{1}{2}\hat{d}\BIGP{q_1,q_2}$.

\smallskip

To see the second half holds, for each $q \in \OP{Bisect}(q_1,q_2)$, we have $\hat{d}(q_1,q) = \hat{d}(q,q_2) = \hat{t}_{\nicefrac{1}{2}}(q_1,q_2) = \frac{1}{2}\hat{d}\BIGP{q_1,q_2}$.
Joining a shortest time-path from $q_1$ to $q$ and a shortest time-path from $q$ to $q_2$, we get a time-path, say, $\MC{C}^\prime$, whose length is $\hat{d}(q_1,q_2)$. Therefore $\MC{C}^\prime \in \OP{STP}(q_1,q_2)$.
\end{proof}


\medskip


\begin{ap_lemma}{\ref{lemma-wavefront-recursion}}
$W_p(q,t)$ is formed by the boundary of
$$
\MC{C}^\circ_p(q,t) \cup \bigcup_{s\colon s\in \MBB{R}^2, \\ \hat{d}_p(q,s) < d_p(q,s) \le t} W_p\BIGP{s,t-\hat{d}_p(q,s)},$$
where $\MC{C}^\circ_p(q,t)$ is the $p$-circle with center $q$ and radius $t$.
\end{ap_lemma}

\begin{proof}[Proof of Lemma~\ref{lemma-wavefront-recursion}]
To see why this lemma holds, it suffices to observe that, for any $s \in \MBB{R}^2$ such that $\hat{d}_p(q,s) = d_p(q,s)$, the wavefront of $\MC{C}^\circ_p(s,t-\hat{d}_p(q,s))$ is dominated by $\MC{C}^\circ_p(q,t)$.
\end{proof}


\medskip


\begin{ap_lemma}{\ref{lemma-wavefront-from-highway}}
For $1<p<\infty$, $v_\MC{H} < \infty$, and a point $q \in \MC{H}$ which we assume to be $(0,0)$ for the ease of presentation, the upper-part of $W_p(q,t)$ that lies above $\MC{H}$ consists of the following three pieces:
\begin{compactitem}
	\item
		$\OP{Fan}_p(q,t)$: the circular-sector of the $p$-circle with radius $t$, ranging from $\BIGP{-\hat{x}_t,y_p(-\hat{x}_t,t)}$ to $\BIGP{\hat{x}_t,y_p(\hat{x}_t,t)}$.
	\item
		$\OP{LT}_p(q,t)$, $\OP{RT}_p(q,t)$: two line segments joining $\BIGP{-v_\MC{H}\cdot t,0}$, $\BIGP{-\hat{x}_t,y_p(-\hat{x}_t,t)}$, and $\BIGP{\hat{x}_t,y_p(\hat{x}_t,t)}$, $\BIGP{v_\MC{H}\cdot t,0}$, respectively. Moreover, $\OP{LT}_p(q,t)$ and $\OP{RT}_p(q,t)$ are tangent to $\OP{Fan}_p(q,t)$.
\end{compactitem}
The lower-part that lies below $\MC{H}$ follows symmetrically.
For $v_\MC{H} = \infty$, the upper-part of $W_p(q,t)$ consists of a horizontal line $y = t$.
\end{ap_lemma}

\begin{proof}[Proof of Lemma~\ref{lemma-wavefront-from-highway}]
By Lemma~\ref{lemma-wavefront-recursion}, it suffices to consider the boundary of
$$\bigcup_{s\colon \MC{H}, \hat{d}_p(q,s) \le t} W_p\BIGP{s,t-\hat{d}_p(q,s)}.$$
That is, the boundary of the union of the $p$-circles $\MC{C}(\gamma)$ with center $\gamma$ and radius $t-\frac{\gamma}{v_\MC{H}}$, where $\gamma$ is a real number with $-v_\MC{H}\cdot t \le \gamma \le v_\MC{H}\cdot t$.
See also Fig.~\ref{fig-on-highway-wavefront-circles}~(a) for an illustration.

\begin{figure*}[h]
\hspace{-20pt}\includegraphics[scale=1.4]{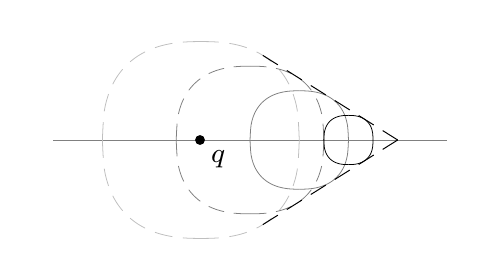}
\hspace{-20pt}\includegraphics[scale=1.4]{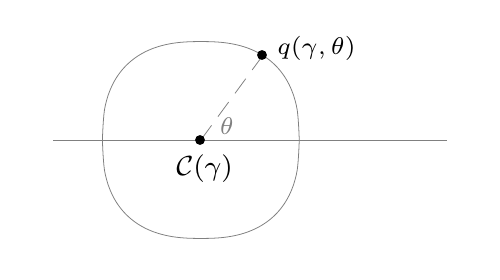}
\caption{The propagation of the $p$-circles along the highway $\MC{H}$, where in this figure $p = 2.5$.}
\label{fig-on-highway-wavefront-circles}
\end{figure*}

First we argue that the upper-right boundary of $$\bigcup_{0 \le \gamma \le v_\MC{H}\cdot t}\MC{C}(\gamma)$$ contains a line segment that is tangent to $\MC{C}(0)$.
For each $\theta$ with $0\le \theta\le \pi/2$ and each $\MC{C}(\gamma)$, let $q(\gamma,\theta)$ denote the specific point of $\MC{C}(\gamma)$ such that the line between $q(\gamma,\theta)$ and the center of $\MC{C}(\gamma)$ forms an angle of $\theta$ above the horizontal line. See also Fig.~\ref{fig-on-highway-wavefront-circles}~(b).

For a fixed $\theta$, the coordinates of $q(\gamma,\theta)$ is given by the following equations.
\begin{align*}
\begin{cases}
	y_q = \BIGP{x_q-\gamma} \cdot\tan\theta \\
	\BIGP{x_q-\gamma}^p+y_q^p = \BIGP{t-\frac{\gamma}{v_\MC{H}}}^p
\end{cases}
\end{align*}
Simplifying the equations, we get
$$y_q = \frac{t-\frac{\gamma}{v_\MC{H}}}{\BIGP{1+\tan^p\theta}^{1/p}}\cdot\tan\theta,$$
which implies that $y_q$ is a linear function of $\gamma$.
This shows that, for fixed $\theta$, the trace of $q(\gamma,\theta)$ as $\gamma$ moves from $0$ to $v_\MC{H}t$ is exactly a line segment with the right end-point $(v_\MC{H}t,0)$.
Since this holds for all $\theta$, and, since we have one-to-one correspondence between $\theta$ and a point on the $p$-circle, union of these segments is exactly the set $\bigcup_{0 \le \gamma \le v_\MC{H}\cdot t}\MC{C}(\gamma)$.
Hence the upper-right boundary will be the upper-most line segment, which is exactly the one tangent to $\MC{C}(0)$.

\smallskip

Below, we show that the corresponding tangent point is given by $\BIGP{\hat{x}_t, y_p(\hat{x}_t,t)}$.
Consider the curve of the $p$-circle $\MC{C}(0)$ in the first quadrant. The equation is given by $y = \BIGP{t^p-x^p}^{1/p}$. Hence the slope of the line tangent at any specific point is given by
$$\frac{dy}{dx} = -\BIGP{t^p-x^p}^{\frac{1-p}{p}}x^{p-1}.$$
The point $(x,y)$ of interest has a tangent line that passes through $(v_\MC{H},0)$. Therefore,
$$\frac{y}{x-v_\MC{H}} = -\BIGP{t^p-x^p}^{\frac{1-p}{p}}x^{p-1}.$$
Simplifying the equation we obtain the point $\BIGP{\hat{x}_t, y_p(\hat{x}_t,t)}$ as claimed.

\medskip

For the case $v_\MC{H} = \infty$, notice that we have
$$\bigcup_{s\colon \MC{H}, \hat{d}_p(q,s) \le t} W_p\BIGP{s,t-\hat{d}_p(q,s)} = \bigcup_{s\in \MC{H}}W_p\BIGP{s,t}.$$
The upper-part of the boundary is exactly the horizontal line $y=t$.
\end{proof}


\medskip


%

\subsection{Regarding the incidence angle, Lemma~\ref{lemma-incidence-angle}.}

\begin{ap_lemma}{\ref{lemma-incidence-angle}}
For any point $q = \BIGP{x_q,y_q}$, $1\le p < \infty$, and $1<v_\MC{H}\le\infty$, if a shortest time-path starting from $q$ uses the highway $\MC{H}$, then it must enter the highway with an incidence angle $\alpha$.
\end{ap_lemma}

\smallskip

%

%
Before proving Lemma~\ref{lemma-incidence-angle}, we need more structural properties regarding the propagation of wavefront under the influence of the highway $\MC{H}$.
Let $q = (0,y_q)$ be a point. First, we argue that, for any $t \ge y_q$ and for all $q^\prime \in \MC{H}$ such that $d_p(q,q^\prime) \le t$, $\OP{Fan}_p\BIGP{q^\prime,t-d_p(q,q^\prime)}$ does not belong to the boundary of $W_p\BIGP{q,t}$.
The trace of each point on $\OP{Fan}_p\BIGP{q^\prime,t-d_p(q,q^\prime)}$ represents a shortest time-path from $q^\prime$ to that point.
Therefore by the triangle inequality, we know that $\OP{Fan}_p\BIGP{q^\prime,t-d_p(q,q^\prime)}$ will be dominated by the trace of the $p$-circle with center $q$.

\smallskip

Below, we prove this lemma for a constrained situation where we have a point $q_0$ on the highway $\MC{H}$ and a point $q$ off the highway.

\begin{lemma}
For $q_0 = (0,0) \in \MC{H}$ and any $q=(x_q,y_q) \in \MBB{R}^2$ such that $x_q \ge 0$ and $y_q > 0$,
the shortest time-path between $q_0$ and $q$ is
(1) the straight line-segment joining $q_0$ and $q$, if $x_q \le y_q\tan\alpha$, or
(2) the straight line-segment from $q$ into highway $\MC{H}$ with incidence angle $\alpha$ toward the direction of $q_0$, followed by a traversal along the highway to $q_0$, if $x_q \ge y_q\tan\alpha$.
\end{lemma}

\begin{proof}
Let $q^\prime = (x_{q^\prime}, y_{q^\prime})$ be a point in $\OP{Bisect}(q_0,q)$. We consider the following cases. 
\begin{enumerate}[(a)]
	\item
		If $q^\prime \in \OP{Fan}_p(q_0,\hat{t}_{\nicefrac{1}{2}}(q_0,q))$, then we know that the trace of $q^\prime$ belongs to the $p$-circle $\MC{C}^\circ(q,\hat{t}_{\nicefrac{1}{2}}(q_0,q))$. By the symmetry of $p$-circles, we know that $\MC{C}^\circ(q,\hat{t}_{\nicefrac{1}{2}}(q_0,q))$ and $\OP{Fan}_p(q_0,\hat{t}_{\nicefrac{1}{2}}(q_0,q))$ share a common tangent line, and the three points $q_0, q, q^\prime$ are co-linear.
		Therefore from the trace of $p$-circles we know that the straight line-segment is a shortest time-path between $q_0$ and $q$.
		Moreover, since the angle $\OP{Fan}_p(q_0,\hat{t}_{\nicefrac{1}{2}}(q_0,q))$ spans is between $\pm \alpha$, we have $x_q \le y_q\tan\alpha$.
		
	\item
		If $q^\prime \in \OP{RT}_p(q_0,\hat{t}_{\nicefrac{1}{2}}(q_0,q)) \backslash \MC{H}$, then $q^\prime$ must come from the $p$-circle $\MC{C}^\circ(q,\hat{t}_{\nicefrac{1}{2}}(q_0,q))$. Moreover, from the perspective of $\OP{RT}_p(q_0,\hat{t}_{\nicefrac{1}{2}}(q_0,q)) \backslash \MC{H}$, by the proof of Lemma~\ref{lemma-wavefront-from-highway}, the trace of the point $q^\prime$ comes from another $p$-circle from a point, say, $q^{\prime\prime} \in \MC{H}$ such that $q^{\prime\prime}$ and $q^\prime$ forms an angle of $\alpha$ to the vertical line.
		Therefore the points $q$, $q^\prime$, and $q^{\prime\prime}$ are again co-linear, and the trace of $q^\prime$ from both sides, which are the line segments $\overline{qq^{\prime\prime}}$ and $\overline{q^{\prime\prime}q_0}$, is a shortest time-path between $q_0$ and $q$.
		
	\item
		Finally, if $q^\prime \in \MC{H}$, then we know that $q^\prime$ does not belong to the $p$-circle $\MC{C}^\circ_p(q,\hat{t}_{\nicefrac{1}{2}}(q_0,q))$ as it will touch $\MC{H}$ to the right of $q^\prime$ before it touches $q^\prime$.
		Therefore, from (a) we know that $q_x > y_q\tan\alpha + x_{q^\prime}$ for otherwise the trace of the $p$-circle $\MC{C}^\circ_p(q,\hat{t}_{\nicefrac{1}{2}}(q_0,q))$ will touch $q^\prime$ directly.
		Hence, in the situation between $q^\prime \in \MC{H}$ and $q \notin \MC{H}$, we know that case (a) will never happen, and we either have case (b) or another recursion.
		After each iteration, the remaining portion of $\MC{H}$ between the two points are shrunk by at least $y_q$. Therefore we will eventually end up with case (b).
\end{enumerate}
This proves the lemma.
\end{proof}

\medskip

For any point $q = (x_q,y_q) \in \MBB{R}^2$ with $y_q \ge 0$, let $q^+_\MC{H}$ and $q^-_\MC{H}$ be two points with coordinates $\BIGP{x_q \pm y_q\tan\alpha,0}$, respectively. 
Below, we characterize the behavior of the shortest time-paths between any pair of points.
Let $q_1 = (x_1,y_1), q_2 = (x_2,y_2) \in \MBB{R}^2$ be two points such that $x_1 \le x_2$ and $y_1,y_2 \ge 0$.

Let $q$ be a point in $\OP{Bisect}(q_1,q_2)$.
If $q \in \MC{H}$, then by the lemma above, we know that there exists shortest time-path that uses the highway $\MC{H}$ with incidence angle $\alpha$. Otherwise, if $q \notin \MC{H}$, we have two cases to consider.
(1) If $q$ belongs to both $\MC{C}^\circ_p(q_1,\hat{t}_{\nicefrac{1}{2}}(q_1,q_2))$ and $\MC{C}^\circ_p(q_2,\hat{t}_{\nicefrac{1}{2}}(q_1,q_2))$, then the straight line-segment joining $q_1$ and $q_2$ is a shortest time-path between $q_1$ and $q_2$.
(2) Otherwise, we know that $q$ belongs to a $p$-circle and a left-tangent $\OP{LT}_p$, which implies the existence of a shortest time-path that uses the highway $\MC{H}$ with incidence angle $\alpha$.

We have considered all possible cases and obtained that, whenever a shortest time-path uses the highway $\MC{H}$, it has an incidence angle $\alpha$. This proves Lemma~\ref{lemma-incidence-angle}.


\medskip




\subsection{The $L_\infty$-metric.}

%

%
Below we discuss the case when $p = \infty$. 
While the approach and also the arguments to deriving the properties for the $L_\infty$ metric are similar to that we use for $1<p<\infty$, we state only the key properties
without repeating 
the technical details.
First, regarding the propagation of the wavefront from a point of the highway $\MC{H}$, we have the following lemma.

\begin{lemma}
For $v_\MC{H} < \infty$, and a point $q \in \MC{H}$ which we assume to be $(0,0)$ for the ease of presentation, the upper-part of $W_\infty(q,t)$ that lies above $\MC{H}$ consists of the following three pieces:
\begin{compactitem}
	\item
		$\OP{Fan}_\infty(q,t)$: the circular-sector of the $\infty$-circle with radius $t$, which is a horizontal line segment, ranging from $\BIGP{-t,t}$ to $\BIGP{t,t}$.
	\item
		$\OP{LT}_\infty(q,t)$, $\OP{RT}_\infty(q,t)$: two line segments joining $\BIGP{-v_\MC{H}\cdot t,0}$, $\BIGP{-t,t}$, and $\BIGP{t,t}$, $\BIGP{v_\MC{H}\cdot t,0}$, respectively.
\end{compactitem}
The lower-part that lies below $\MC{H}$ follows symmetrically.
For $v_\MC{H} = \infty$, the upper-part of $W_\infty(q,t)$ consists of a horizontal line $y = t$.
\end{lemma}

\begin{lemma}
When $v_\MC{H} < \infty$, if a shortest time-path under the $L_\infty$ metric uses the highway $\MC{H}$, then it must enters the highway with an incidence angle $\alpha = \pi/4$.
When $v_\MC{H} = \infty$, a shortest time-path that uses the highway $\MC{H}$ may enter the highway with an incidence angle $-\pi/4 \le \alpha \le \pi/4$.
\end{lemma}

Also refer to Fig.~\ref{fig-wavefront-linfty} for an illustration on the propagation of wavefronts and the behavior of shortest time-paths under the $L_\infty$ metric.

\begin{figure*}[h]
{\includegraphics[scale=1.2]{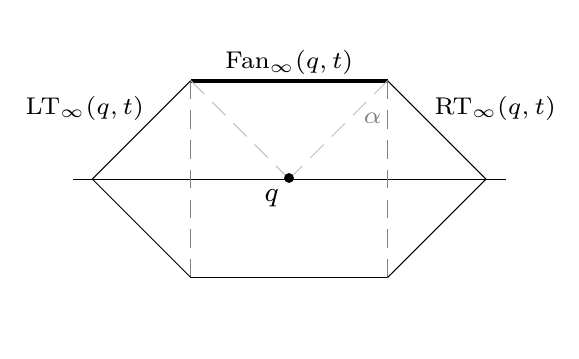}}
{\includegraphics[scale=1.2]{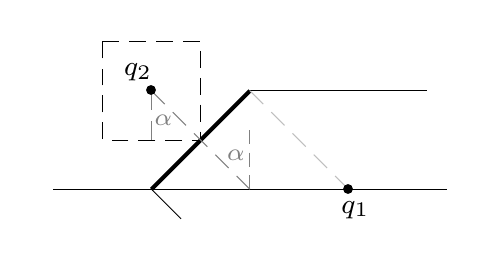}}
\caption{(a) The wavefront $W_\infty(q,t)$ from a point $q \in \MC{H}$. (b) The shortest time-path in $L_\infty$ when the highway is used.}
\label{fig-wavefront-linfty}
\end{figure*}


\medskip

\subsection{Walking Regions}

\begin{ap_lemma}{\ref{lemma-wr-general}}
For any $p$ with $1\le p\le \infty$ and any point $q = (x_q,y_q)$ with $y_q \ge 0$, 
$\OP{WR}_p(q)$ is characterized by the following two curve:
(a) \emph{right discriminating curve}, which is the curve $q^\prime = (x^\prime,y^\prime)$ satisfying $x^\prime\ge x_q + y_q\tan\alpha$ and
\begin{align}
\BIGC{\overline{qq^\prime}}_p = \BIGC{\overline{qq^+_\MC{H}}}_p + \BIGC{\overline{q^\prime q^{\prime-}_\MC{H}}}_p + \frac{1}{v_\MC{H}}\BIGC{\overline{q^+_\MC{H}q^{\prime-}_\MC{H}}}_p.
\label{eq-wr-general}
\end{align}
(b) The \emph{left discriminating curve} is symmetric with respect to the line $x=x_q$. 
\end{ap_lemma}

\begin{proof}[Proof of Lemma~\ref{lemma-wr-general}]
This lemma follows directly from Lemma~\ref{lemma-incidence-angle}. Deriving the equation of the curve of the walking-region for $p$ with $1\le p<\infty$ is straight-forward. Describing the boundary of the walking-region for $L_\infty$, however, is a bit more tricky.
Below we discuss in more detail.

\smallskip

%
%

%

The two lines with slopes $\pm 1$ that pass through $q$ divide the plane into four regions. Let $\OP{I}$, $\OP{II}$, and $\OP{III}$ denote the three specific regions that lie to the left of $q$ from top to bottom. 
Also refer to Fig.~\ref{fig-wr-general-l-infty}~(a) for an illustration.

For points $q^\prime =(x^\prime, y^\prime)\in \OP{I}$, we have $d_\infty(q,q^\prime) = y^\prime-y_q$, while $d_\infty(q,\MC{H}) + d_\infty(q^\prime,\MC{H}) = y^\prime+y_q$. This implies any shortest time-path between $q$ and $q^\prime$ will never enter $\MC{H}$.
Therefore $\OP{I} \subseteq \OP{WR}_\infty(q)$.
Similarly we have $\OP{II} \subseteq \OP{WR}_\infty(q)$.

Consider a point $q^\prime = (x^\prime,y^\prime) \in \OP{III}$, we have $d_\infty(q,q^\prime) = x_q-x^\prime$. Expanding Eq.~(\ref{eq-wr-general}) we get
$$\BIGP{x_q-x^\prime} = y_q + y^\prime + \frac{1}{v_\MC{H}}\BIGP{(x_q-y_q) - (x^\prime+y^\prime)},$$
which simplifies to 
\begin{equation}
y^\prime = x_q-y_q -x^\prime. 
\label{eq-wr-linfty}
\end{equation}
Fig.~\ref{fig-wr-general-l-infty}~(b) shows the boundary of $\OP{WR}_\infty(q)$ for a point $q \in \MBB{R}^2$.

\begin{figure*}[h]
\begin{minipage}{0.4\textwidth}
\centering
\includegraphics[scale=1]{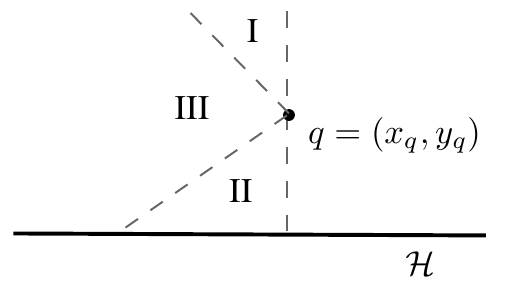}
\end{minipage}
\begin{minipage}{0.4\textwidth}
\centering
\includegraphics[scale=1]{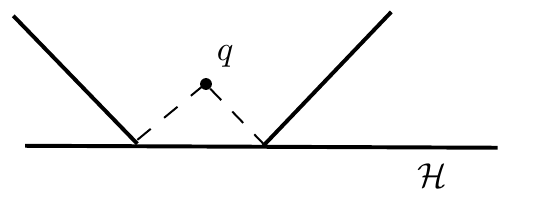}
\end{minipage}
\caption{(a) The three regions partitioned by the lines of slope $\pm 1$. (b) The walking region of a point $q \in \MBB{R}^2$ under $L_\infty$.}
\label{fig-wr-general-l-infty}
\end{figure*}

\end{proof}


\medskip


\begin{ap_lemma}{\ref{lemma-wr-dominance}}
Let $q_1 = (x_1,y_1)$ and $q_2 = (x_2, y_2)$ be two points such that $x_1 \le x_2$. If $y_1 \ge y_2$, then $\OP{WR}_{\ell}(q_2)$ lies to the right of $\OP{WR}_{\ell}(q_1)$.
Similarly, if $y_1 \le y_2$, then $\OP{WR}_{r}(q_1)$ lies to the left of $\OP{WR}_{r}(q_2)$.
\end{ap_lemma}

\begin{proof}[Proof of Lemma~\ref{lemma-wr-dominance}]

First we prove this lemma for the $p=1$ and $p=\infty$.
Assume that $y_1 \ge y_2$. 
\begin{itemize}
	\item
		For $p = 1$, by Proposition~\ref{prop-l1-previous}, it suffices to consider the relative position between the two lines characterizing the discriminating curve for $q_1$ and $q_2$, respectively.
		First, since $x_2 \ge x_1$, the line specified in Proposition~\ref{prop-l1-previous}~(b) for $q_2$ lies to the right of that for $q_1$.
		Similarly, since $y_2 \le y_1$, we have $\BIGP{x_2-y_2/\beta} - \BIGP{x_1 - y_1/\beta} \ge 0$, which implies that the vertical strip for $q_2$ also lies to the right of that for $q_1$.
		This implies that $\OP{WR}_{\ell}(q_2)$ lies to the right of $\OP{WR}_{\ell}(q_1)$.
		
	\item
		For $p = \infty$, by Eq.~\ref{eq-wr-linfty} from the proof of Lemma~\ref{lemma-wr-general}, we know that the left discriminating curve for a point $q = (x_q,y_q)$ is characterized by the linear equation $y = x_q-y_q-x$.
		Since we have $x_1-y_1 \le x_2-y_2$ by our assumption, we know that the left discriminating curve for $q_1$ lies to the left of that for $q_2$.
		
\end{itemize}
%
The other direction of this lemma follows from a symmetric argument.

\medskip

Below we prove for the case $1<p<\infty$.
To this end, observe that, it suffices to show that for two points $q = (0,y_q)$ and $q^\prime=(0,y_{q^\prime})$ such that $y_q \ge y_{q^\prime}$, $\OP{WR}_\ell(q^\prime)$ lies to the right of $\OP{WR}_\ell(q)$ and $\OP{WR}_r(q^\prime)$ lies to the left of $\OP{WR}_r(q)$.

\medskip

For the case when $v_\MC{H} = \infty$, we prove that, 
for any point $q_w^\prime \in \OP{WR}_{\ell}(q^\prime)$, except for the origin $(0,0)$, there always exists a point $q_w \in \OP{WR}_{\ell}(q)$ with the same $y$-coordinate such that $q_w$ lies to the left of $q_w^\prime$.
By Lemma~\ref{lemma-wr-general}, any point of $\OP{WR}_{\ell}(q^\prime)$ is given by $\OP{LT}(q^{\prime-}_{\MC{H}},t) \cap \MC{C}^\circ(q^\prime,y_{q^\prime}+t)$ for some $t \ge 0$, and there is a one-to-one correspondence.
%

Let $t_c$, $t_c \ge 0$, be the specific moment such that $$\OP{LT}(q^{\prime-}_{\MC{H}},t_c) \cap \MC{C}^\circ(q^\prime,y_{q^\prime}+t_c) = q_w^\prime,$$ 
and let 
$$q_w = \OP{LT}(q_{\MC{H}}^-,t_c) \cap \MC{C}^\circ(q,y_q+t_c).$$
We know that $q_w^\prime$ and $q_w$ share the same $y$-coordinate as they belong to $\OP{LT}(q^{\prime -}_{\MC{H}},t_c)$ and $\OP{LT}(q_{\MC{H}}^-,t_c)$, which are both given by the equation $y = t_c$.
Since $y_q \ge y_{q^\prime}$, we know that $\MC{C}^\circ(q^\prime,y_{q^\prime}+t_c)$ is completely contained by $\MC{C}^\circ(q,y_q+t_c)$,
which implies that $q_w$ lies to the left of $q_w^\prime$. This proves that $\OP{WR}_\ell(q^\prime)$ lies to the right of $\OP{WR}_\ell(q)$ for the case $v_\MC{H} = \infty$.
See also Fig.~\ref{fig-wr-dominance-infty} for an illustration.

\begin{figure*}[htp]
\centering
\includegraphics[scale=0.9]{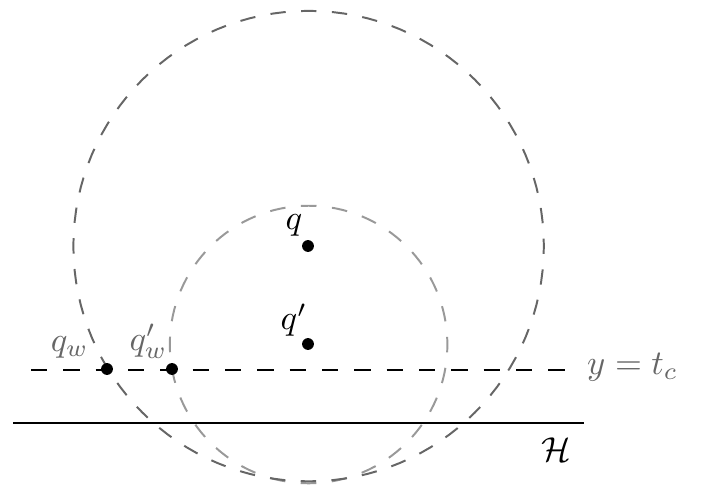}
\caption{An illustration of the relative positions between $q_w$ and $q_w^\prime$ for $v_\MC{H} = \infty$.}
\label{fig-wr-dominance-infty}
\end{figure*}

\medskip

For the case $1 < v_\MC{H} < \infty$, 
we argue that for any point $q_w^\prime \in \OP{WR}_{\ell}(q^\prime) \backslash \MC{H}$, there always exists a point $q_w \in \OP{WR}_{\ell}(q)$ such that $q_w$ lies strictly to the lower-left of $q_w^\prime$.
By Lemma~\ref{lemma-wr-general}, any point of $\OP{WR}_{\ell}(q^\prime) \backslash \MC{H}$ is given by $\OP{LT}_{p}(q^{\prime -}_\MC{H},t) \cap \MC{C}^\circ\BIGP{q^\prime, d_p(q^\prime,q^{\prime -}_\MC{H}) + t}$ for some $t > 0$, and there is a one-to-one correspondence.
Let $t_c$, where $t_c > 0$, be the specific moment such that
$$\OP{LT}(q^{\prime -}_\MC{H},t_c) \cap \MC{C}^\circ\BIGP{q^\prime, d(q^\prime,q^{\prime -}_\MC{H}) + t_c} = q_w^\prime,$$
and let 
$$q_w = \OP{LT}(q^{-}_\MC{H},t_c) \cap \MC{C}^\circ\BIGP{q, d(q,q^{-}_\MC{H}) + t_c}.$$

Furthermore, since both $q^-_\MC{H}$ and $q^{\prime-}_\MC{H}$ belong to $\MC{H}$, it is not difficult to see that $\OP{W}(q^{-}_\MC{H},t_c)$ is a horizontal translation of $\OP{W}(q^{\prime -}_\MC{H},t_c)$ to the left-hand side by $d(q^{-}_\MC{H}, q^{\prime -}_\MC{H})$.

\smallskip

To help present our proof, let $\overline{\OP{LT}}(q^{\prime-}_\MC{H},t_c)$ and $\overline{\OP{LT}}(q^{-}_\MC{H},t_c)$ denote the lower-part of $\OP{W}(q^{\prime-}_\MC{H},t_c)$ and $\OP{W}(q^-_\MC{H},t_c)$ that correspond to ${\OP{LT}}(q^{\prime-}_\MC{H},t_c)$ and ${\OP{LT}}(q^{-}_\MC{H},t_c)$.
Moreover, let $\overline{q}^{\prime}$ and $\overline{q}$ denote the right endpoint of $\overline{\OP{LT}}(q^{\prime-}_\MC{H},t_c)$ and $\overline{\OP{LT}}(q^{-}_\MC{H},t_c)$, respectively.
See also Fig.~\ref{fig-on-highway-wavefront-lower} for an illustration.
\begin{figure*}[t]
\centering
\includegraphics[scale=1.3]{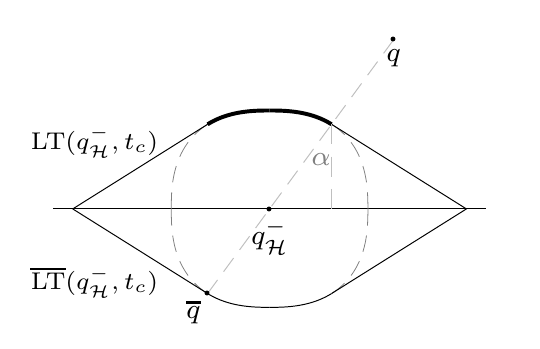}
\caption{An illustration of the relative positions between $q$, $q^-_\MC{H}$, and $\overline{q}$.}
\label{fig-on-highway-wavefront-lower}
\end{figure*}


%
From Lemma~\ref{lemma-wr-general}, the definition of $\alpha$, and the definition of $q^{\prime-}_\MC{H}$, we know that $q^\prime$, $q^{\prime-}_\MC{H}$, and $\overline{q}^\prime$ are collinear.
Therefore, we know that $\overline{\OP{LT}}(q^{\prime-}_\MC{H},t_c)$ is tangent to $\MC{C}^\circ\BIGP{q^\prime, d(q^\prime,q^{\prime -}_\MC{H}) + t_c}$ at $\overline{q}^\prime$.
Similarly, $\overline{\OP{LT}}(q^{-}_\MC{H},t_c)$ is tangent to $\MC{C}^\circ\BIGP{q, d(q,q^{-}_\MC{H}) + t_c}$ at $\overline{q}$.
\begin{figure*}[htp]
\centering
\includegraphics[scale=0.9]{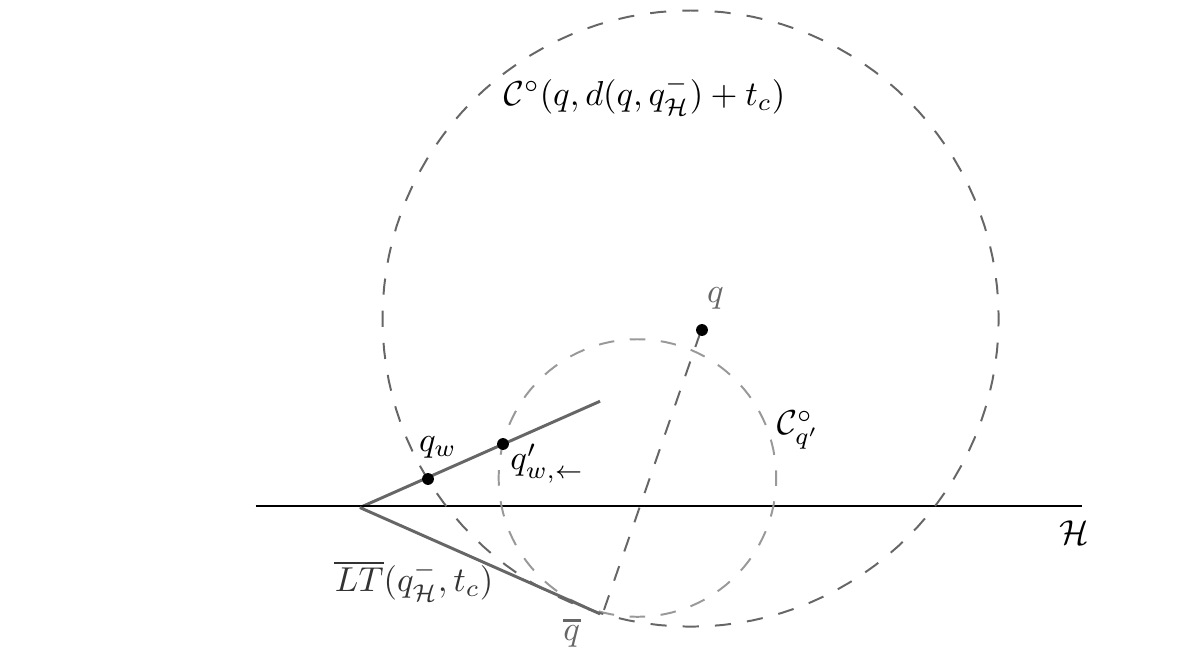}
\caption{An illustration on the translated images of $\MC{C}^\circ\BIGP{q^\prime, d(q^\prime,q^{\prime-}_\MC{H})+t_c}$ and $q^\prime_w$, for which the latter is denoted by $q^\prime_{w,\leftarrow}$.}
\label{fig-wr-dominance-finite}
\end{figure*}

Now consider the horizontal translation of $\MC{C}^\circ\BIGP{q^\prime, d(q^\prime,q^{\prime -}_\MC{H}) + t_c}$ to the left-hand side by $d(q^{-}_\MC{H}, q^{\prime -}_\MC{H})$. Let $\MC{C}^\circ_{q^\prime}$ be the translated image.
Since $\overline{\OP{LT}}(q^{\prime-}_\MC{H},t_c)$ is tangent to $\MC{C}^\circ\BIGP{q^\prime, d(q^\prime,q^{\prime -}_\MC{H}) + t_c}$ at $\overline{q}^\prime$, we know that $\overline{\OP{LT}}(q^{-}_\MC{H},t_c)$ is also tangent to $\MC{C}^\circ_{q^\prime}$ at $\overline{q}$, which implies that $\MC{C}^\circ_{q^\prime}$ is completely contained within $\MC{C}^\circ\BIGP{q, d(q,q^{-}_\MC{H}) + t_c}$.
This shows that the translated image of $q^\prime_w$ lies strictly to the upper-right of $q_w$, since they both belong to $\OP{LT}(q^-_\MC{H},t_c)$. This implies that $q_w$ lies to the lower-left of $q^\prime_w$. See also Fig.~\ref{fig-wr-dominance-finite} for an illustration.

\medskip

The other half, that $\OP{WR}_r(q^\prime)$ lies to the left of $\OP{WR}_r(q)$, can be proved by a symmetric argument. This proves the lemma.

\end{proof}


\medskip

\subsection{Closure and Time-Convex Hull of a Point Set.}

\begin{ap_lemma}{\ref{lemma-closure-pair-of-points}}
Let $q_1, q_2 \in \MBB{R}^2$ be two points, when the highway is not used, the set of all shortest time-paths between $q_1$ and $q_2$, i.e., the closure of the set $\MB{S}=\BIGBP{q_1,q_2}$ is:
\begin{compactitem}
	\item
		The smallest bounding rectangle of $\MB{S}$, when $p = 1$.
	
	\item
		The straight line segment $\overline{q_1q_2}$ joining $q_1$ and $q_2$, when $1<p<\infty$.
		
	\item
		The smallest bounding parallelogram whose slopes of the four sides are $\pm 1$, i.e., a rectangle rotated by $45^\circ$, that contains $\MB{S}$.
\end{compactitem}
\end{ap_lemma}

\begin{proof}[Proof of Lemma~\ref{lemma-closure-pair-of-points}]
Without loss of generality, assume that $q_1 = (x_1,y_1)$ and $q_2 = (x_2,y_2)$ such that $x_1\le x_2$ and $y_1 \le y_2$.

\smallskip

For the case $p = 1$, observe that both the paths $q_1 \rightarrow (x_1,y_2) \rightarrow q_2$ and $q_1 \rightarrow (x_2,y_1) \rightarrow q_2$ are shortest time-paths between $q_1$ and $q_2$ since their lengths are equal to $d_1(q_1,q_2)$.
Therefore, the closure of $\MB{S} = \BIGBP{q_1,q_2}$ contains every shortest time-path between $(x_1, y)$ and $(x_2, y)$ for $y_1 \le y \le y_2$. This gives the bounding box of $\MB{S}$, for which we denote by $B_1(\MB{S})$ for brevity.
Furthermore, it is not difficult to see that, for each $q, q^\prime \in B_1(\MB{S})$, $\OP{Bisect}(q,q^\prime) \subseteq B_1(\MB{S})$, since the unit-circles in $L_1$ are diamonds.
Therefore the closure of $\MB{S}$ is $B_1(\MB{S})$.

\smallskip

Since the unit-circles of $L_\infty$, squares, are exactly the unit-circles of $L_1$ being rotated by $45^\circ$, we obtain the claimed result for $p=\infty$ by rotating the coordinate system by $45^\circ$.

\smallskip

For $1<p<\infty$, first observe that, any tangent line of a $p$-circle intersects exactly one point. Furthermore, there exists one-to-one correspondence between the slope of a line and a pair of extreme-points from the $p$-circle whose joined line segment passes through the center.
Therefore, when the two $p$-circles with center $q_1$ and $q_2$ meet for the first time, they adopt a common tangent line.
%
This shows that, $q_1$, $q_2$, and the point at which the two $p$-circles meet are co-linear.
The same conclusion holds for further bisecting points. Therefore the shortest time-path between $q_1$ and $q_2$ is the straight line-segment between them, and it is unique.
This shows that the straight line-segment is also the closure of $\MB{S} = \BIGBP{q_1,q_2}$.
\end{proof}


\medskip



\section{Constructing the Time-Convex Hull}

\subsection{Problem Complexity}

\begin{ap_lemma}{\ref{lemma-lower-bound}}
$y_0\BIGP{p,\epsilon}$ is well-defined for all $p$ with $1\le p\le \infty$ and all $\epsilon > 0$. Furthermore, the answer to the minimum gap problem on $\MC{I}$ is ``yes'' if and only if the number of clusters in the time-convex hull of $\MB{S}$ is less than $n$.
\end{ap_lemma}

\begin{proof}[Proof of Lemma~\ref{lemma-lower-bound}]
First we argue that $y_0(p,\epsilon)$ is well-defined. Let $v_\MC{H} = \infty$. Since for all value of $p$, the right discriminating curve, $f_{q(y)}$, has a one-on-one correspondence with each value of $y$, $f^-1_{q(y)}$ is well-defined and continuous, spanning value from $0$ to infinity as $y$ goes from $0$ to infinity. Therefore for all possible $\epsilon$, we can always find a $y_0(p,\epsilon)$.
The correctness of this reduction follows directly as testing whether or not a point belongs to the walking region of another is equivalent to the testing of the gap-length.
\end{proof}


\medskip




\subsection{An Optimal Algorithm}

\paragraph{Regarding the Point Inclusion Test.}

An approach has been proposed to resolve this step efficiently, e.g., Yu and Lee~\cite{Yu:2007:TCH:1270397.1271493}, and Aloupis et al.~\cite{Aloupis:2010:HHR:1598101.1598676}.
Below we sketch how it efficiently locates the leftmost cluster that has to be merged with the incoming point $q_i$.

\smallskip

The main idea is to store $\OP{WR}_r(C_j)$ for each cluster $C_j$ created so far by a linked-list, followed by using a linear scan on these lists to retrieve the smallest index for each $q_i$.
%
The correctness of this approach 
relies primarily on the following facts:
\begin{compactitem}
	\item
		For each $1\le j, j^\prime \le k$, $j\ne j^\prime$, $\OP{WR}_r(C_j)$ intersects $\OP{WR}_r(C_{j^\prime})$ at most once.
	\item
		If $q_i \in \OP{WR}_r(C_j)$, then all the clusters in between, i.e., $C_{j+1}, \ldots, C_k$, have to be merged as well in order to satisfy the closure property of time-convexity.
\end{compactitem}
What really helps is the \emph{dominance property} implied by the above two properties: Let $C_j$ and $C_{j^\prime}$ be two clusters such that (1) $C_j$ lies to the left of $C_{j^\prime}$ and (2) $\OP{WR}_r(C_j)$ intersects $\OP{WR}_r(C_{j^\prime})$ at a point $q$. Then the part of $\OP{WR}_r(C_{j^\prime})$ to the right of $q$ is not important any more as it lies above $\OP{WR}_r(C_j)$ and whenever a point belongs to $\OP{WR}_r(C_j)$, $C_{j^\prime}$ has to be merged as well.
This allows us to concentrate on the rightmost boundary of the walking-regions and enables a smooth transition when a new point arrives.

\smallskip

Let $\OP{WR}_r\BIGP{\MC{C}}$ denote the rightmost boundary of the union of all walking-regions of clusters in $\MC{C}$.
During the algorithm, we maintain a set of ordered linked-lists where in each linked-list we store a partial boundary of the walking-regions that belonged to $\OP{WR}_r\BIGP{\MC{C}}$ at a certain moment.
When a new point $q_i$ arrives, we walk the linked-list $\OP{WR}_r\BIGP{\MC{C}}$ to locate the piece that lies above $\OP{WR}_r(q_i)$.
This specific piece of $\OP{WR}_r\BIGP{\MC{C}}$, denoted $\hat{\OP{WR}}_r\BIGP{\MC{C}}$, is then replaced by the partial piece of $\OP{WR}_r(q_i)$ that lies below.
The replaced piece, $\hat{\OP{WR}}_r\BIGP{\MC{C}}$, is inserted to the tail of the ordered linked-lists we maintained as a new list.
Then we walk the ordered linked-lists, starting from the tail list, to check if $q_i$ belongs to the walking-region of any cluster and to locate the smallest index if it exists.
As the new point always comes in $x$-monotonic order, line segments and parabolae that are expired with respect to $q_i$ will be discarded during the walk.
This walk is continued until we reach a list whose boundary lies above $q_i$, i.e., it intersects with the vertical ray shooting upward from $q_i$.
%
We omit the remaining details and refer the reader to the work of~\cite{Yu:2007:TCH:1270397.1271493,Aloupis:2010:HHR:1598101.1598676}.

%

\bigskip

%


\begin{ap_theorem}{\ref{thm-opt-algo}}
Provided that the one-sided segment sweeping query can be answered in $Q(n)$ time using $P(n)$ preprocessing time and $S(n)$ storage, the time-convex hull for a given set $\MB{S}$ of $n$ points under the given $L_p$-metric can be computed in $\MC{O}\BIGP{n\log n + nQ(n) + P(n)}$ time using $\MC{O}\BIGP{n+S(n)}$ space.
\end{ap_theorem}

\begin{proof}[Proof of Theorem~\ref{thm-opt-algo}]
It remains to argue that 
the time complexity is $\MC{O}\BIGP{n\log n + nQ(n) + P(n)}$.
Observe that each cluster-merging operation results in at most one newly-created edge. Therefore at most $\MC{O}(n)$ edges is identified and inserted into $\MB{E}$. Furthermore, after each segment query invoked by the algorithm,
either one edge is removed from $\MB{E}$, or the number of clusters is decreased by one. Therefore at most $\MC{O}(n)$ query is performed. This proves this theorem.
\end{proof}





%


%

\begin{ap_lemma}{\ref{lemma-segment-query-pruning}}
Let $\MC{I}$ be an interval, $\MC{C}\colon \MC{I} \rightarrow \MBB{R}$ be a convex function, i.e., $\MC{C}\BIGP{\frac{1}{2}\BIGP{x_1+x_2}} \ge \frac{1}{2}\BIGP{\MC{C}(x_1)+\MC{C}(x_2)}, \forall x_1,x_2 \in \MC{I}$, that is differentiable almost everywhere,
$\OP{L}$ 
be a segment with slope $\theta_L$, $-\infty < \theta_L < \infty$, and $q=\BIGP{x_q,\MC{C}(x_q)}$ be a point on the curve $\MC{C}$ such that 
$$\lim_{x\rightarrow x_q^-}\frac{\OP{d}\MC{C}(x)}{\OP{d}x} \ge \theta_L \ge \lim_{x\rightarrow x_q^+}\frac{\OP{d}\MC{C}(x)}{\OP{d}x}.$$
If $q$ lies under $\overleftrightarrow{\OP{L}}$, then the curve $\MC{C}$ never intersects $\OP{L}$.
\end{ap_lemma}

\begin{proof}[Proof of Lemma~\ref{lemma-segment-query-pruning}]
Notice that, at any point $x$ where $\MC{C}(x)$ is differentiable, if $\frac{\OP{d}\MC{C}}{\OP{d}x}\BIGP{x} > \theta_L$, then $\MC{C}(x)$ tends to go above $\OP{L}$ as $x$ increases.
Similarly, if $\frac{\OP{d}\MC{C}}{\OP{d}x}\BIGP{x} < \theta_L$, then $\MC{C}(x)$ tends to go below $\OP{L}$ as $x$ increases.
By the convexity of $\MC{C}$, $\forall x_l,x_r \in \MC{I}$ such that $x_l < x_r$, 
we have 
\begin{align}
\lim_{x\rightarrow x_l^+}\frac{\OP{d}\MC{C}}{\OP{d}x}\BIGP{x} \ge \lim_{x\rightarrow x_r^-}\frac{\OP{d}\MC{C}}{\OP{d}x}\BIGP{x}.
\label{ieq-slope-curve}
\end{align}

Suppose that $q$ lies under $\overleftrightarrow{\OP{L}}$ and the curve $\MC{C}$ intersects $\OP{L}$ at a point $q^\prime = \BIGP{x^\prime,\MC{C}(x^\prime)}$. 
If $x_q < x^\prime$, then we know that at some $x^{\prime\prime}$ with $x_q < x^{\prime\prime} \le x^\prime$, $\frac{\OP{d}\MC{C}(x)}{\OP{d}x} > \theta_L \ge \lim_{x\rightarrow x_q^+}\frac{\OP{d}\MC{C}}{\OP{d}x}\BIGP{x}$, which is a contradiction to Ineq.~(\ref{ieq-slope-curve}) and the fact that $x_q < x^\prime$.
For the other case where $x^\prime < x_q$, we obtain a similar contradiction as well.
\end{proof}

\end{appendix}


\end{document}